\documentclass[runningheads,a4page]{llncs}

\usepackage{amsmath}
\usepackage{amssymb}

\usepackage{amsthm}
\usepackage{mathrsfs}
\usepackage{wasysym}
\usepackage{color}
\usepackage{xspace}
\usepackage{txfonts}
\usepackage{graphicx}
\usepackage{enumerate}
\usepackage{tikz}
\usepackage{nicefrac}

\newcommand{\ltl}{{LTL}\xspace}
\newcommand{\ctl}{{CTL}\xspace}
\newcommand{\ctlst}{CTL*\xspace}

\newcommand{\pctl}{{PCTL}\xspace}
\newcommand{\pctlst}{{PCTL*}\xspace}
\newcommand{\mutl}{{$\mu$TL}\xspace}
\newcommand{\pmutl}{{P$\mu$TL}\xspace}
\newcommand{\papa}{{PAPA}\xspace}

\newcommand{\sat}{\textsc{Satisfiability}\xspace}
\newcommand{\emp}{\textsc{Emptiness}\xspace}

\newcommand{\up}{\textbf{UP}\xspace}
\newcommand{\coup}{\textbf{co-UP}\xspace}
\newcommand{\np}{\textbf{NP}\xspace}
\newcommand{\conp}{\textbf{co-NP}\xspace}
\newcommand{\pspace}{\textbf{PSPACE}\xspace}
\newcommand{\exptime}{\textbf{EXPTIME}\xspace}

\newcommand{\calA}{\mathcal{A}}
\newcommand{\calD}{\mathcal{D}}
\newcommand{\calZ}{\mathcal{Z}}

\newcommand{\opX}{\mathsf{X}}
\newcommand{\opU}{\mathsf{U}}
\newcommand{\opR}{\mathsf{R}}
\newcommand{\opF}{\mathsf{F}}
\newcommand{\opG}{\mathsf{G}}
\newcommand{\funcD}{\mathbf{D}}

\newcommand {\depth}[1]{||{#1}||}
\newcommand {\pmat}[1] {\emph{\textbf{#1}}}
\newcommand {\lang}{\mathscr{L}}
\newcommand {\winning}{\mathscr{W}}
\newcommand {\binder}{\mathscr{D}}
\newcommand {\semantics}[1]{\llbracket{#1}\rrbracket}

\DeclareMathOperator*{\cyl}{cyl}
\DeclareMathOperator*{\proj}{proj}
\DeclareMathOperator*{\prob}{prob}
\DeclareMathOperator*{\cls}{cls}

\begin{document}
\title{A Simple Probabilistic Extension of Modal Mu-calculus}
 \author{ Wanwei Liu\inst{1}
        \and Lei Song\inst{2}
        \and Ji Wang\inst{1}
        \and Lijun Zhang\inst{3}
 }

\institute
{
  \inst{}
School of Computer Science, \\National University of Defense Technology, Changsha, P. R. China
 \and
 \inst{}
Center of Quantum Computation and Intelligence Systems, \\University of Technology, Sydney, Australia
   \and
  \inst{} 
  State Key Laboratory of Computer Science, \\Institute of Software, Chinese Academy of Sciences   
}

\maketitle

\begin{abstract}
    Probabilistic systems are an important theme in AI domain.
    As the specification language, \pctl
    is the most frequently used logic for reasoning about
    probabilistic properties. In this paper, we present a natural and succinct
    probabilistic extension of $\mu$-calculus, another prominent logic in
    the concurrency theory. We study the relationship with
    \pctl. Surprisingly, the expressiveness is highly orthogonal
    with \pctl. The proposed logic captures some useful
    properties which cannot be expressed in \pctl.
    We investigate the model checking and satisfiability problem, and show that the model
    checking problem is in \up$\cap$\coup, and  the satisfiability
    checking can be decided via reducing into solving parity games. This is in contrast to PCTL as well,
    whose satisfiability checking is still an open problem.
\end{abstract}

\section{Introduction}
\label{Sec: Introduction}

Temporal logics are heavily used in theoretical computer science and AI-related fields.
Among those, modal $\mu$-calculus receives a lot of attraction ever since Kozen's seminal work \cite{Koz83}.
See for example, \cite{BB87,Kat98,Wal00,Ber02}.
 Moreover, various temporal logics including \ltl \cite{Pnu77}, \ctl
 \cite{EC80}, \ctlst \cite{EH86} are extensively studied.
 It is known that their expressiveness  is strictly less \cite{Dam95} than $\mu$-calculus (aka. \mutl),
and their model checking algorithm has been proposed: for \ctl the problem can be solved in polynomial time,
 whereas for \ltl the problem is \pspace-complete \cite{SC85}.

Probabilistic systems, such as Markov chains and Markov decision processes, are an important
theme in AI domain.
To reason about properties for probabilistic systems, the logic \ctl was first extended with probabilistic quantifiers in
\cite{HJ94} , resulting in the logic \pctl.
Intuitively, $(a \opU^{\geq 0.9} b)$ means that the probability of reaching $b$-states along $a$-states is at least $0.9$.
At the same time, probabilistic \ltl  and its extension \pctlst have all been studied.
As in the classical setting, model checking problem for \pctl can be solved in polynomial time, whereas only exponential algorithms are known for \ltl \cite{CSS03}.
There have also been several attempts to extend \mutl with probabilities in the literature. As we shall discuss in the related work, the extensions are either highly
non-trivial in terms of the complexity of the corresponding model checking and satisfiability problems, or hindered from the restriction of
fixpoint nesting.

We propose a natural and succinct extension of \mutl in this paper, and name it \pmutl.
The logic is acquired by equipping the next operator with probability quantifiers, and keeping other parts
as standard \mutl.
We have for instance the formula $\nu Z.(a\wedge\opX^{\geq 0.8}Z)$.
We investigate the model checking, expressiveness,
and satisfiability problems of \pmutl.

In detail, we first investigate the model checking problem of \pmutl upon Markov chains.
It turns out to be a straightforward adaptation of the classical
algorithms for \mutl, and the complexity remains in \up $\cap$ \coup.
We then give a comprehensive study on the expressiveness of \pmutl by comparing with \pctl,
and prove that \pmutl is orthogonal with \pctl in expressiveness. However, for the qualitative
fragments (i.e., probabilities may appear in a formula are only $0$ and $1$),
we show that qualitative \pmutl is strictly more expressive (w.r.t. finite Markov chains).
On the other side, the satisfiability
checking is quite challenging: we exploit the notion of probabilistic alternating parity automata (\papa, for short),
and reduce the \sat problem into the \emp problem of \papa.
Further, this is reduced to solving parity games, and it is shown that both of these two problems are in 2\exptime.
This is in contrast to \pctl as well,
whose \sat checking is still an open problem (cf. \cite{BFKK08,BFS12}).

\paragraph{An illustrating example}
We introduce a running example to motivate our work: Suppose there is
a hacker trying to attack a remote server. The hacker has
a supercomputer at hand and is trying to guess the password in a brute-force manner. For
simplicity, we assume the password is a sequence of $l$ letters, each
of which is from `0'-`9', `a'-`z', and `A'-`Z'. Therefore, the total number of
possible passwords is $n=62^l$. The hacker let the supercomputer
randomly generate a password, and see whether the decryption
succeeds. If yes, the hacker wins; otherwise he
tries with another one. However, if the
supercomputer generates three wrong passwords in a row, it will be
blocked for a certain amount of time until it can start another round
of attacking | assuming that the password may be changed during the
blocked moment, hence it does not make sense for the supercomputer to store all
generated passwords.
The whole process is illustrated in Fig.~\ref{Fig: Hacker}.
Starting from $s_1$, we can see that the
probability of eventually reaching $\mathit{attacked}$, i.e., the hacker
decrypts successfully, equal 1, no matter how
big $l$ is  (hence, the \pctl formula $\opF^{\geq 1}attacked$ holds),
and we may conclude that the system is unsafe | this is of
course against our intuition, as such system is considered to be safe if $l$ is big enough.
However, as we will show later, all
\pctl formulae are not capable of expressing this
property. By making use of \pmutl, such property of security can be characterized
easily as follows: $\nu Z.(\neg\mathit{attacked}\land\opX^{\ge p}Z))$ with
$p=\nicefrac{n-3}{n-2}$, where $\neg\mathit{attacked}$ denotes all other
states in Fig.~\ref{Fig: Hacker} different from $s_5$.

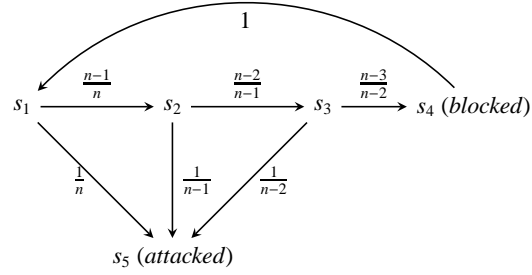
\begin{figure}[!t]
  \centering
\scalebox{1}{
 \begin{tikzpicture}[->,>=stealth,auto,node distance=2cm,semithick,scale=1,every node/.style={scale=1}]
	\tikzstyle{state}=[minimum size=20pt,circle,draw,thick]
	\tikzstyle{stateNframe}=[]
	every label/.style=draw
        \tikzstyle{blackdot}=[circle,fill=black, minimum size=6pt,inner sep=0pt]
     \node[stateNframe](s0){$s_1$};
     \node[stateNframe](s1)[right of=s0]{$s_2$};
     \node[stateNframe](s2)[right of=s1]{$s_3$};
     \node[stateNframe](s3)[right of=s2]{$s_4\ (\mathit{blocked})$};
     \node[stateNframe](ss)[below of=s1]{$s_5\ (\mathit{attacked})$};
     \path (s0) edge node{$\frac{n-1}{n}$} (s1)
                edge node[left]{$\frac{1}{n}$} (ss)
           (s1) edge node{$\frac{n-2}{n-1}$} (s2)
                edge node{$\frac{1}{n-1}$} (ss)
           (s2) edge node{$\frac{n-3}{n-2}$} (s3)
                edge node[right]{$\frac{1}{n-2}$} (ss)
           (s3) edge[bend right=45] node{1} (s0);

\end{tikzpicture}}
  \caption{An illustration of the hacking process}
  \label{Fig: Hacker}\vspace*{-1em}
\end{figure}


\paragraph{Motivation from AI perspective} The presented logic has
the following potential application in AI domain:
\begin{itemize}
  \item First of all, Markov chains and Markov decision processes are the
     basic models in several areas of AI. As a logic with semantics defined
     w.r.t. such models, it could definitely be used in designating probability-relevant
     properties upon them. Particularly, the properties that could not be expressed by \pctl.
  \item Motion planing is an important topic in AI area, where
    standard \mutl has once been adopted \cite{bhatia2011motion}, because of its powerful expressiveness and
    the decidability of its \sat problem.
Thus, we expect that \pmutl could be used in stochastic motion planning | since,
    \pmutl is a decidability-preserving extension of \mutl.
  \item Fixpoints play an important role in mathematics and
    computer science. In AI area, it is used to designate non-terminating behaviors of intelligent systems,
    such as maintenance goals \cite{DBLP:journals/logcom/Singh98}.
    Fixpoints act as the elementary ingredients in \pmutl, hence such logic can also be used
    in such a situation.
\end{itemize}

\paragraph{Related work}
Probabilistic extensions of \mutl have been studied by
many authors:
e.g., $\mu$-calculi proposed in \cite{MM97,HK97,AM01,MM02,MM07,Mio12b}
interpret a formula as a function from states to real values in $[0,1]$, whose
semantics is different from \pmutl. A further extension of
$\mu$-calculus was proposed in~\cite{Mio12a}, which is able to
encode the full \pctl. However, the model checking and \sat algorithms are
still unknown for these calculi and are ``far from
trivial''~\cite{Mio12b}. The other probabilistic $\mu$-calculus was
introduced in~\cite{CIN05} along with a model checking
algorithm for it. Moreover, it is able to encode \pctl formulae as
well. However, that calculus  only allows
alternation-free formulae (cf. \cite{EL86}).

Very recently | and independently |, Castro, Kilmurray, and Piterman
present another extension by adding fixpoints to full \pctl
\cite{CKP15}. The calculus they introduced is more expressive than
logics \pctl and \pctlst. Moreover, it is also easy to see that it is
a proper super logic of our logic \pmutl as well. They show the model
checking problem is in \np$\cap$\conp.  We note that some examples
in our paper are similarly investigated in
\cite{CKP15}. Since the logic in \cite{CKP15} subsumes \pctl, its \sat
problem is also left open. However in this paper we show \sat
of \pmutl could be reduced to solving parity games, which makes
this problem solvable in 2\exptime.


\section{Preliminaries}

In this paper, we fix a countable set $\calA$ of \emph{atomic propositions}, ranging over $a, b, a_1$ etc,
and fix a countable set $\calZ$ of \emph{formula variables}, ranging over $Z, Z_1$ etc.

A \emph{Markov chain} is a tuple $M=(S,\pmat{T},L)$, where
 $S$ is a finite set of \emph{states};
 $\pmat{T}: S\times S\to[0,1]$ is the matrix of transition-probabilities, fulfilling
$\sum_{s'\in S}\pmat{T}(s,s')=1$ for every $s\in S$; and
 $L: S\to 2^\calA$ is the labeling function.
 A \emph{pointed Markov chain} is
a pair $(M,s)$ where $M$ is a Markov chain $(S,\pmat{T},L)$ and $s\in S$ is the \emph{initial state}.

An (infinite) \emph{path} $\pi$ of $M$ is an infinite sequence of states
$s_0,s_1,\ldots$, such that  $s_i\in S$ and $\pmat{T}(s_i,s_{i+1})>0$ for
each $i$.
A basic \emph{cylinder} $\cyl(s_0,s_1,\ldots,s_n)$ of $M$ is the set of infinite
paths having $s_0,s_1,\ldots, s_n$ as the prefix.

According to the standard theory of Markov process, the pointed Markov chain $(M,s)$
uniquely derives a \emph{measure space} $(\Pi_{M,s},\Delta_{M,s},\prob_{M,s})$ where
$\Pi_{M,s}$ consists of all infinite paths of $M$; $\Delta_{M,s}$ is the
minimal Borel field containing all basic cylinder of $M$
(i.e., $\Delta_{(M,s)}$ is closed under complementation and countable intersection);
and the measuring function $\prob_{M,s}$ fulfills:
$\prob_{M,s}(\cyl(s_0,s_1,\ldots,s_n))$ equals $0$ if $s\neq s_0$, and equals $\prod_{i<n}\pmat{T}(s_i,s_{i+1})$ otherwise.
We say a set $P\subseteq\Pi_{M,s}$ is \emph{measurable} if $P\in\Delta_{M,s}$.
\cite{Var85} shows that the intersection of $\Pi_{M,s}$
and an omega-regular set must be measurable.

The syntax of  \pctl formulae is described by the following abstract grammar:
\begin{displaymath}
f \Coloneqq  \top\mid \bot\mid a \mid \neg a \mid \opX^{\sim p} f \mid f\wedge f\mid f\vee f\mid
f\opU^{\sim p} f \mid f\opR^{\sim p} f
\end{displaymath}
where $\sim\in\{>,\geq\}$ and $p\in[0,1]$.
We also abbreviate $\top\opU^{\sim p} f$ and $\bot\opR^{\sim p} f$ as
$\opF^{\sim p} f$ and $\opG^{\sim p} f$, respectively.

Semantics of a \pctl formula is given w.r.t. a Markov chain.
For each \pctl formula $f$ and a Markov chain $M=(S,\pmat{T},L)$,
we will use $\semantics{f}_M$ to denote the subset of $S$ satisfying $f$,
inductively defined as follows.
\begin{itemize}
  \item $\semantics{\top}_M=S$;  $\semantics{\bot}_M=\emptyset$.
  \item $\semantics{a}_M=\{s\in S\mid a\in L(s)\}$; $\semantics{\neg a}_M=\{s\in S\mid a\not\in L(s)\}$.
  \item $\semantics{\opX^{\sim p}f}_M=\{s\in S \mid
  \sum_{s'\in\semantics{f}_M}\pmat{T}(s,s')\sim p\}$.
  \item $\semantics{f_1\wedge f_2}_M=\semantics{f_1}_M\cap\semantics{f_2}_M$;
   $\semantics{f_1\vee f_2}_M=\semantics{f_1}_M\cup\semantics{f_2}_M$.
  \item $\semantics{f_1\opU^{\sim p}f_2}_M=\{s\in S\mid
  \prob_{M,s}\{\pi\in\cyl(s)\mid \pi\models f_1\opU f_2\}\sim p\}$
  and $\semantics{f_1\opR^{\sim p}f_2}_M=\{s\in S\mid
  \prob_{M,s}\{\pi\in\cyl(s)\mid \pi\models f_1\opR f_2\}\sim p\}$.
\end{itemize}
In addition, for an infinite path $\pi=s_0,s_1,\ldots$ of $M$, the notation
$\pi\models f_1\opU f_2$ stands for that there is some $i\geq 0$ such
that $s_i\in\semantics{f_2}_M$ and $s_j\in\semantics{f_1}_M$ for each $j<i$.
Meanwhile, $\pi\models f_1\opR f_2$ holds if either $\pi\models
f_2\opU(f_1\wedge f_2)$ or $s_j\in\semantics{f_2}_M$ for each $j$.
To simplify notations, in what follows we denote by $M,s\models f$ whenever
$s\in\semantics{f}_M$ holds.


\section{\pmutl, Syntax and Semantics}
\label{Sec: Logic}

In this section we present a simple probabilistic extension of modal $\mu$-calculus,
called \pmutl.
The syntax of \pmutl formulae is depicted as follows:
\begin{displaymath}
f \Coloneqq \top\mid \bot\mid a\mid\neg a \mid Z\mid \opX^{\sim p} f \mid f\wedge f \mid f \vee f \mid \mu Z. f \mid \nu Z. f
\end{displaymath}

Semantics of a \pmutl formula is given w.r.t. a Markov chain
$M=(S,\pmat{T},L)$ and an assignment $e:\calZ\to 2^S$.
Similarly, for each \pmutl formula $f$,
we denote by $\semantics{f}_M(e)$ the state set satisfying
$f$ under $e$. Inductively:
\begin{itemize}
\item $\semantics{\top}_M(e)=S$ and $\semantics{\bot}_M(e)=\emptyset$.
\item $\semantics{a}_M(e) = \{s\in S\mid a\in L(s)\}$ and $\semantics{\neg a}_M(e)=\{s\in S\mid a\not\in L(s)\}$.
\item $\semantics{Z}_M(e)=e(Z)$.
\item $\semantics{\opX^{\sim p}f}_M(e)=\{s\in S\mid \sum_{s'\in\semantics{f}_M(e)}\pmat{T}(s,s')\sim p\}$.
\item $\semantics{f_1\wedge f_2}_M(e)=\semantics{f_1}_M(e)\cap\semantics{f_2}_M(e)$\\
and $\semantics{f_1\vee f_2}_M(e)=\semantics{f_1}_M(e)\cup\semantics{f_2}_M(e)$.
\item $\semantics{\mu Z. f}_M(e)=\bigcap\{S'\subseteq S\mid \semantics{f}_M(e[Z\mapsto S'])\subseteq S'\}$
and	$\semantics{\nu Z. f}_M(e)=\bigcup\{S'\subseteq S\mid \semantics{f}_M(e[Z\mapsto S'])\supseteq S'\}$.
\end{itemize}

Indeed, $\semantics{\mu Z.f}_M(e)$ (resp. $\semantics{\nu Z.f}_M(e)$)
could be computed as in the classical setting via the following iteration:
\begin{enumerate}
\item let $S_0=\emptyset$ (resp. $S_0=S$);
\item subsequently, let $S_{i+1}=\semantics{f}_M(e[Z\mapsto S_i])$;
\item stops if $S_{\ell+1} = S_\ell$, and returns $S_\ell$.
\end{enumerate}
Note that the algorithm obtains a monotonic chain with such an iteration,
and hence it must terminate within finite steps.
Actually, $\semantics{\mu Z.f}_M(e)$ (resp. $\semantics{\nu Z.f}_M(e)$) captures
the least (resp. greatest) solution of $X=\semantics{f}_M(e[Z\mapsto X])$ within $2^S$.

Semantical definition of \pmutl formulae also
yields the model checking algorithm.

\begin{theorem}
 \label{Thm: Model-checking}
 The model checking problem of \pmutl is in \up$\cap$\coup.
\end{theorem}

Indeed, the proof is analogous to the non-probabilistic version \cite{Jur98,Wil00}
and the only noteworthy difference lies from handling $\opX^{\sim p}$-
subformulae, opposing to $\Box$- and $\Diamond$- subformulae, which
could be proceeded in (deterministic) polynomial time.

In what follows, we directly denote by $\semantics{f}_M$ in the case that $f$ is a closed formula
(i.e., each variable of $f$ is bound),
 and we also denote by $M,s\models f$
if $s\in\semantics{f}_M$.

Below we give some example properties:
\begin{enumerate}[(1)]
    \item The formula $\nu Z.(a\wedge\opX^{> 0.8}Z)$ describes that there
  exists an $a$-region,
	where each state has less than $0.2$ probability to escape from it immediately
	(i.e., in one step).
    \item  $\nu Z.(a\wedge\opX^{>0}\opX^{>0}Z)$ says that there is a cycle in the Markov chain,
	such that $a$ holds at least in every even step.
    \item  $M,s\models \mu Z.(a\vee\opX^{\geq 0.6}Z)$ if some $a$-state is reachable
     from $s$, but at each step, one just has  some probability (not less than $0.6$) to go
     on with the right direction.
    \item
    The \pmutl formula  $\mu Z.(b\vee(a\wedge\opX^{\geq 1} Z))$ holds if $a\opU b$ holds along each path.
    It is  stronger than the  property described by the \pctl
    formula $a \opU^{\geq 1}b$. For the latter allows the existence of $a$-cycles.
    \item As  a more complicated example, the formula $\nu Z_1.(a\vee\mu Z_2.(a\vee\opX^{>0}Z_2)\wedge\opX^{\geq 1}Z_1)$
	just tells the story that ``$a$ will be surely encountered'',
        as described by $\opF^{\geq 1}a$ with \pctl.
 \end{enumerate}

Given a \pmutl formula $f$ and a bound variable $Z$, we use $\binder_f(Z)$ to
denote the subformula which binds $Z$ in $f$. For example, let $f=\mu
Z_1(a\wedge\nu Z_2.(b\wedge\opX^{>=0.3}Z_2)\vee\opX^{>0.6}Z_1)$, then
we have $\binder_f(Z_1)=f$ and $\binder_f(Z_2)= \nu
Z_2.(b\wedge\opX^{\geq 0.3}Z_2)$.

We say that a \pmutl formula $f$ is \emph{guarded}, if
the occurrence of each bound variable $Z$ in $\binder_f(Z)$ is in
the scope of some $\opX$-operator.
The following theorem could be proven in a same manner as that in \cite{Wal00}.

\begin{theorem}
\label{Thm: Guarded}
For each \pmutl formula $f$, there is a guarded  formula $f'$
such that $\semantics{f'}_M(e)=\semantics{f}_M(e)$ for every $M$ and $e$.
\end{theorem}

Thus, in what follows, we always assume that each
\pmutl formula is guarded.


\section{Expressiveness}

In this section, we will give a comparison between \pmutl and \pctl,
and we are only concerned about closed \pmutl formulae.
For a  \pmutl formula $f$ and a \pctl formula $g$, we say that
$f$ and $g$ are \emph{equivalent} if $\semantics{f}_M = \semantics{g}_M$
for every Markov chain $M$, denoted as $f\equiv g$.

First of all,  we will show that some \pmutl formula
could not be equivalently expressed by any \pctl formula.

\begin{theorem}
\label{Thm: PmuTL-beyond}
Let $f=\nu Z.(a\wedge\opX^{\geq 0.5} Z)$,
then $g\not\equiv f$ for every \pctl formula $g$.
\end{theorem}

\begin{proof}

To show this, we need first construct two families of Markov chains, namely,
$M_0,M_1,\ldots,$ and $M'_0,M'_1,M'_2,\ldots$.
\par
For the first group, let $M_n=(\{s_0,s_1,\ldots,s_n\},\pmat{T}_n,L_n)$, where:
$\pmat{T}_n(s_0,s_0)=1$ and $\pmat{T}_n(s_{i+1},s_i)=1$ for each $i<n$
(hence $\pmat{T}_n(s_i,s_j)=0$ for any other $s_i$, $s_j$).
In addition, $L_n(s_0)=\emptyset$ and $L_n(s_i)=\{a\}$ for each
$0<i\leq n$.

 For the second ones, let $M'_n=(\{s'_0,s'_1,\ldots,s'_n\},\pmat{T}'_n,L'_n)$
where:
$\pmat{T}'_n(s'_n,s'_n)=\pmat{T}'_n(s'_n,s'_{n-1})=0.5$,
$\pmat{T}'_n(s_0,s_0)=1$, and $\pmat{T}'_n(s'_{i+1},s'_i)=1$ for every $i<n-1$.
In addition, $L'_n(s'_0)=\emptyset$ and $L'_n(s'_i)=\{a\}$ for each
$0<i\le n$.

Given a \pctl formula $g$, let $N(g)$ be the
maximal nesting depth of temporal-operators of $g$.
According to  \cite[Thm. 10.45]{BK08}, we have that
$M'_n,s'_n\models g$ if and only if $M_n,s_n\models g$
whenever $n\geq N(g)$.

Observe the fact that $M'_n,s'_n\models f$ and $M_n,s_n\not\models f$
for every $n\geq 1$.
Assume that there exists some \pctl formula $g$ fulfilling $f\equiv g$,
then we have
$$\begin{array}{rlcl}
 & M'_{N(g)},s'_{N(g)}\models f & \Longleftrightarrow &
 M'_{N(g)},s'_{N(g)}\models g \\
 \Longleftrightarrow &
 M_{N(g)},s_{N(g)}\models g & \Longleftrightarrow & M_{N(g)},s_{N(g)}\models f
\end{array}$$
and hence it results in a contradiction.
\end{proof}



Conversely, the following theorem reveals that there also exists
some \pctl formula that could not be equivalently expressed by any \pmutl
formula.

\begin{theorem}
\label{Thm: PCTL-beyond}
Let $f =\opF^{\geq 0.5}a$, then $g\not\equiv f$ for every (closed) \pmutl
formula $g$.
\end{theorem}

\begin{proof}
 Let $M=(\{s_1,s_2,s_3\},\pmat{T}, L)$ be the (family of) Markov chain(s) where:
 $L(s_1)=L(s_2)=\emptyset$, $L(s_3)=\{a\}$, $\pmat{T}(s_1,s_1)=x,
 \pmat{T}(s_1,s_2)=y, \pmat{T}(s_1,s_3)=z$, and $\pmat{T}(s_2,s_2)=\pmat{T}(s_3,s_3)=1$,
 with $x,y,z\in(0,1)$ and $x+y+z=1$.
 \par For every \pctl and/or closed \pmutl formula $g$, we let $P_x(g)$ be the
 proposition that ``for the fixed $x$, there are infinitely many $y$ making $M,s_1\models g$ and
 there are infinitely many $y$
 making $M,s_1\not\models g$''.
 We now show that if $g$ is a closed \pmutl formula, then there exists some
 $x_g<1$ such that $P_x(g)$ does not hold whenever $x\in(x_g,1)$.
\begin{itemize}
\item Such $x_g$ can be arbitrarily chosen if $g=\bot$, $g=\top$, $g=a$ or $g=\neg a$.
\item In the case that $g=g_1\wedge g_2$, assume by contradiction that such $x_g$ does not exist,
	then it implies that for every $x\in(0,1)$, there exists some $x'>x$ such that $P_{x'}(g)$ holds.
	Observe that $M,s_1\models g$ implies both $M,s_1\models g_1$ and $M,s_1\models g_2$;
	and $M,s_1\not\models g$ implies either $M,s_1\not\models g_1$ or $M,s_1\not\models g_2$.
	Thus, we can infer that either $x_{g_1}$ or $x_{g_2}$ does not
        exist, which violates the induction hypothesis.
\item Proof for the case of $g=g_1\vee g_2$ is similar to the above.
\item If $g=\opX^{\sim p} g'$ and $p\in(0,1)$, whenever $x\in (\max\{p,1-p\},1)$, since $\sim\in\{>,\geq\}$,
	then $M,s_1\models g$ iff $M,s_1\models g'$ because $y+z<p$ in such situation.
	In this case, we may just let $x_g=\max\{x_{g'}, p,1-p\}$.
\item If $g=\opX^{\geq 1} g'$, then we need to distinguish two cases: 1) There exist $x,y\in(0,1)$ such that
	$M,s_1\models g$ holds, then we can immediately infer that both $M,s_2\models g'$ and $M,s_3\models g'$.
	In addition, observe that truth values of $g'$ on $s_2$ and $s_3$ are irrelevant  to $x$ and $y$.
	It implies that in such case $M,s_1\models g$ iff $M,s_1\models g'$, and hence, we may just let $x_g=x_{g'}$.
	2) There is no such $x$ and $y$ having $M,s_1\models g$ holds, in such situation, $x_g$ can be any
	number in $(0,1)$.
\item If $g=\opX^{>0}g'$, then the proof is similar to the above.
\item When $g=\opX^{\geq 0}g'$ (or $g=\opX^{>1}g'$), things would be trivial, because $g$ could be reduced to
	$\top$ (resp. $\bot$) in such case.
\item  If $g=\mu Z. g'$, we let $g_0=\bot$ and $g_{i+1}=g'[Z/g_i]$. Since that $M$ is a $3$-state Markov chain,
	then $g$ and $\bigvee_{i\leq 3}g_i$ share the same truth value at every state of $M$.
	This indicates that all least fix-points could be eliminated w.r.t. such Markov chain.
\item When $g=\nu Z.g'$, the preprocessing is almost similar, but we just replace $g$ with $\bigwedge_{i\leq 3}g_i$
	where $g_0=\top$.
\end{itemize}
Now, for the \pctl formula $f=\opF^{\geq 0.5}a$, such $x_f$ does not exist, because, for every $x\in(0,1)$ we have:
$M,s_1\models f$ provided that $y\in [(1-x)/2,1)$; and $M,s_1\not\models f$ if $y\in(0,(1-x)/2)$.
This implies that $P_x(f)$ holds for every $x\in(0,1)$,
and hence $f$ cannot be equally expressed by any \pmutl formula.
\end{proof}

Note that the value $0.5$ in the previous two theorems
can  be generalized to any other probability
$p\in(0,1)$.

We also provide a comparison on the \emph{qualitative} fragments of
\pctl and \pmutl. Probabilities occurring in such fragments
can only be $0$ or $1$.

\begin{theorem}
\label{Thm: Quali-Containing}
Every qualitative \pctl formula can be equally expressed by a qualitative \pmutl
formula.
\end{theorem}

\begin{proof}
We will give a constructive translation procedure,
which takes a qualitative \pctl formula $g$ and outputs
 an equivalent qualitative \pmutl formula $\widetilde{g}$.
Inductively:

\begin{enumerate}
\item $\widetilde{g}=\bot$ if $g=\bot$, or its root operator  is
$\opX^{>1}$, $\opU^{>1}$ or $\opR^{>1}$;\\
 $\widetilde{g}=\top$ if $g=\top$, or its root operator  is
$\opX^{\geq 0}$, $\opU^{\geq 0}$ or $\opR^{\geq 0}$.
\item $\widetilde{g}=\widetilde{g_1}\wedge\widetilde{g_2}$ if $g=g_1\wedge g_2$;
and $\widetilde{g}=\widetilde{g_1}\vee\widetilde{g_2}$ if $g=g_1\vee g_2$.
\item $\widetilde{g}=\opX^{>0}\widetilde{g'}$ if $g=\opX^{>0}g'$;
	and $\widetilde{g}=\opX^{\geq 1}\widetilde{g'}$ if $g=\opX^{\geq 1}g'$.
\item $\widetilde{g}=\mu
Z.(\widetilde{g_2}\vee(\widetilde{g_1}\wedge\opX^{>0}Z))$ if
$g=g_1\opU^{>0}g_2$; \\ and $\widetilde{g}=\nu
Z.(\widetilde{g_2}\wedge(\widetilde{g_1}\vee\opX^{\geq 1}Z))$ if
$g=g_1\opR^{\geq 1}g_2$.
\item $\widetilde{g}=\nu
Z.(\widetilde{g_2}\vee(\widetilde{g_1}\wedge\widetilde{\opF^{>0}g_2}\wedge\opX^{\geq
1}Z)) = \nu Z.(\widetilde{g_2}\vee(\widetilde{g_1}\wedge\mu
Z'.(\widetilde{g_2}\vee\opX^{>0}Z')\wedge\opX^{\geq 1}Z))$ if $g=g_1\opU^{\geq
1}g_2$; \\
and $\widetilde{g}=\mu
Z.(\widetilde{g_2}\wedge(\widetilde{g_1}\vee\widetilde{\opG^{\geq
1}g_2}\vee\opX^{>0}Z)) = \mu Z.(\widetilde{g_2}\wedge(\widetilde{g_1}\vee\nu
Z'.(\widetilde{g_2}\wedge\opX^{\geq1}Z')\vee\opX^{>0}Z))$ if
$g=g_1\opR^{>0}g_2$.
\end{enumerate}
The proof of equivalence could be done by induction on the structure
of the formula.
\end{proof}

Note that Thm. \ref{Thm: Quali-Containing} holds  because we are only
concerned about finite models in this paper.
Interested readers may show that it is not true for infinite Markov chains.

\begin{theorem}
\label{Thm: Quali-Beyond}
The qualitative \pmutl formula $f=\nu Z.(a\wedge\opX^{>0}\opX^{>0}Z)$ cannot
be expressed in qualitative \pctl.
\end{theorem}

\begin{proof}
Construct a series of Markov chains $M''_2, M''_3,\ldots$ such that each
$M''_n$ is the Markov chain $(\{s''_{0},s''_1,\ldots,s''_n\},\pmat{T}''_n,L''_n)$, where
$\pmat{T}''_n(s''_0,s''_0)=1$  and  $\pmat{T}''_n(s''_{i+1},s''_i)=1$ for each $i<n$.
In addition, $L''_n(s''_i)=\{a\}$ for each $i\neq 1$, and $L''_n(s''_1)=\emptyset$.

For a given \pctl formula $g$, let $\hat{g}$ be the \ltl formula obtained from
$g$ by discarding all probability quantifiers, e.g.,
we have $\hat{g}=a\opU(b\vee\opG\neg a)$ if $g=a\opU^{\geq 0.3}(b\vee\opG^{>0.6}\neg a)$.
Since that from $s''_n$ the Markov chain $M''_n$
has exactly one infinite path $\pi_n=s''_n,\ldots,s''_1,(s''_0)^\omega$,
then for each $n\geq 2$ we have $M''_n,s''_n\models g$ if and only if
$\pi_n\models\hat{g}$.
It is shown in \cite{Wol83} that $M''_n,s''_n\models\hat{g}$ iff
$M''_{n+1},s''_{n+1}\models\hat{g}$ in the case of
$n\geq N'(\hat{g})=N'(g)$, where $N'(g)$ and $N'(\hat{g})$ are the nesting depth
of $\opX$-operator of $g$ and $\hat{g}$, respectively.
Thus, we have $M''_n,s''_n\models g$ iff
$M''_{n+1},s''_{n+1}\models g$ in such situation.
This implies that $\nu Z.(a\wedge\opX^{>0}\opX^{>0}Z)$ has no equivalent
qualitative \pctl expression, because we cannot simultaneously have
$M''_n,s''_n\models f$ and $M''_{n+1},s''_{n+1}\models f$ for each $n\geq 2$.
\end{proof}

Note that the conclusion of Thm. \ref{Thm: Quali-Beyond}
is also pointed out in \cite{CIN05}, and we here provide a detailed proof.
Indeed, this proof also works for general \pctl formulae, and hence
the property $\nu Z.(a\wedge\opX^{>0}\opX^{>0}Z)$ even cannot be expressed by any \pctl
 formula.


\section{Automata Characterization}
\label{Sec: Automata}

In this section, we will define a new type of automata recognizing
(pointed) Markov chains,
called \emph{probabilistic alternating parity automata} (\papa, for short),
and such automata could be viewed as the probabilistic extension of
those defined in \cite{Wil00}.

A \papa  $A$
is a tuple $(Q,q_0,\delta,\Omega)$ where: $Q$ is a finite set of \emph{states},
$q_0\in Q$ is the \emph{initial state}, $\delta$ is the \emph{transition
function} to be defined later, and $\Omega: Q\leadsto\mathbb{N}$, is a partial
function of \emph{coloring}; in what follows, we say a state is \emph{colored}
if $\Omega$ is defined for the state.

The notion of \emph{transition conditions} over $Q$ is inductively defined as
follows:
\begin{enumerate}
  \item $\bot$ and $\top$ are transition conditions over $Q$.
  \item For every $a\in\calA$, the literals $a$ and $\neg a$ are transition
  conditions over $Q$.
  \item If $q\in Q$, then $q$ is a transition condition over $Q$.
  \item If $q\in Q$ and $p\in[0,1]$, then $\Circle^{\sim
  p}q$ is a transition condition over $Q$, where $\sim\in\{\geq,>\}$.
  \item If $q_1, q_2\in Q$ then both $q_1\vee q_2$ and $q_1\wedge q_2$ are
  transition conditions over $Q$.
\end{enumerate}

The transition function $\delta$ assigns each state $q\in Q$ a transition
condition over $Q$.

We denote by $R_A$  the \emph{derived graph} of $A$,
its vertex set is just $Q$, and there is an edge from $q_1$ to $q_2$
iff $q_2$ appears in $\delta(q_1)$.
We say that $A$ is \emph{well-structured}, if for every path
$q_1,q_2,\ldots,q_n$ that forms a cycle (i.e., $q_1=q_n$)
in $R_A$, we have that: 1) there
exists some $1\leq i< n$ such that $\delta(q_i)=\Circle^{\sim p}q_{i+1}$ with
some $p\in[0,1]$; 2) there exists some $1\leq j<n$ such that $q_j$ is
colored.
In what follows, we are only concerned about well-structured \papa.

Given a pointed Markov chain $(M,s_0)$ with $M=(S,\pmat{T},L)$ and
$s_0\in S$, a \emph{run} of $A$ over $(M,s_0)$ is a $Q\times S$-labeled
tree $(T,\lambda)$ fulfilling: $\lambda(v_0)=(q_0,s_0)$ for the root vertex
$v_0$; and for each internal vertex $v$ of $T$ with $\lambda(v)=(q,s)$ we require
that
\begin{itemize}
 \item $\delta(q)\neq\bot$, and if $\delta(q)=\top$ then $v$ has no child;  	
  \item $a\in L(s)$ if $\delta(q)=a$, and $a\not\in L(s)$ if $\delta(q)=\neg a$;
  \item if $\delta(q)=q_1\wedge
  q_2$ then $v$ has two children $v_1$ and $v_2$ respectively having
  $\lambda(v_1)=(q_1,s)$ and $\lambda(v_2)=(q_2,s)$;
  \item  if $\delta(q)=q_1\vee q_2$ then $v$ has one child $v'$ with
  $\lambda(v')\in\{(q_1,s),(q_2,s)\}$;
  \item $v$ has one child $v'$ having $\lambda(v')=(q',s)$, if
  $\delta(q)=q'$;
  \item  if $\delta(q)=\Circle^{\sim p}q'$ then $v$ has a set of children
  $v_1,\ldots,v_n$ such that $\lambda(v_i)=(q',s_i)$, where 
  $\sum_{i=1}^n\pmat{T}(s,s_i)\sim p$.
\end{itemize}
For an infinite branch $\tau =v_0,v_1,\ldots$ of $T$, let $n_\tau$
be the number
$$\max\{~n \mid \text{there are infinitely many } i \text{ s.t. }
\Omega(\proj\nolimits_1(\lambda(v_i)))=n\}$$
where $\proj_1(q,s)=q$.
A run $(T,\lambda)$ is \emph{accepting} if $n_\tau$ is an even number,
for every infinite branch $\tau$ of $T$.
A pointed Markov chain $(M,s_0)$ is \emph{accepted} by $A$ if $A$ has
an accepting run over it. We denote by $\lang(A)$ the set consisting of
pointed Markov chains accepted by $A$.

\begin{theorem}
\label{Thm: Formula2Automaton}
Given a closed \pmutl formula $f$, there is a \papa $A_f$ such that:
$M,s\models f$ iff $(M,s)\in\lang(A_f)$, for each pointed Markov chain
$(M,s)$.
\end{theorem}

\begin{proof}
We just let
$A_f=(Q_f,q_f,\delta_f,\Omega_f)$, where:
\begin{itemize}
  \item $Q_f=\{q_g\mid g \textrm{ is a subformula of }f\}$, and hence $q_f\in
  Q_f$;
  \item $\delta_f$ is defined as follows:
	\begin{itemize}
	  \item $\delta_f(q_\bot)=\bot$ and $\delta_f(q_\top)=\top$;
	  \item $\delta_f(q_a)=a$ and $\delta_f(q_{\neg a}) = \neg a$;
	  \item $\delta_f(q_{g_1\wedge g_2})=q_{g_1}\wedge q_{g_2}$
	  and $\delta_f(q_{g_1\vee g_2})=q_{g_1}\vee q_{g_2}$;
	  \item $\delta_f(q_{\opX^{\sim p}g})=\Circle^{\sim p} q_g$;
	  \item $\delta_f(q_{\mu Z.g})=q_g$ and $\delta_f(q_{\nu Z.g})=q_g$;
	  \item $\delta_f(q_Z)=q_{\binder_f(Z)}$.
	 \end{itemize}
  \item $\Omega_f$ is defined at every state $q_Z$ with $Z\in\calZ$ fulfilling:
  	If $Z$ is a $\mu$-variable (resp. $\nu$-variable), then
  	$\Omega_f(q_Z)$ is the minimal odd (resp. even) number which is greater than
  	every $\Omega_f(q_{Z'})$ such that $\binder_f(Z')$ is a subformula of
  	$\binder_f(Z)$.
\end{itemize}

It could be directly examined that $A_f$ is well-structured since $f$ is guarded.
The proof of equivalence can be similarly done as that in \cite{Wil00} | the
only different induction step is to deal with transitions being of $\Circle^{\sim p}
q$ (in that paper, the corresponding cases are $\Box q$ and $\Diamond q$).
Actually, we can see that if a \papa $(Q,q,\delta,\Omega)$ corresponds to the
\pmutl formula $g$, then the \papa
$(Q\cup\{q'\},q',\delta[q'\mapsto\Circle^{\sim p}q],\Omega)$ must correspond to
$\opX^{\sim p}g$.
\end{proof}

\section{Satisfiability Decision}
\label{Sec: Decision}

It is known from Section \ref{Sec: Automata}  that
the \sat problem of \pmutl could be reduced to the \emp problem of \papa.
In this section, we will further  reduce it to parity game solving.

A \emph{parity game} $G$ is a tuple $(V, E, C)$, where: $V$
is a finite set of \ \emph{locations}, and
$V$ could be partitioned into two disjoint sets $V^0$ and $V^1$;
$E\subseteq V\times V$ is the set of  \emph{moves},
required to be total;
and $C:V\leadsto\mathbb{N}$ is a partial function of \emph{coloring},
and we say a location $v$ is \emph{colored}, if $C(v)$ is defined.
In addition, for the game $G$, we require that each loop involves at least one
colored location.

Two players | player $0$ and player $1$, are
respectively in charge of $V^0$ and $V^1$ when $G$ is being played.
A \emph{play} of $G$ starting from $v_0\in V$ is an infinite sequence of
locations $v_0,v_1,\ldots$ made by player $0$ and player $1$ |
for every $i\in\mathbb{N}$, the location
$v_{i+1}$ is chosen by player $0$ (resp. player $1$)
with $(v_i,v_{i+1})\in E$ whenever $v_i\in V^0$ (resp. $v_i\in V^1$).

Player $0$ (resp. player $1$)
\emph{wins} the play $v_0,v_1,\ldots$ if the maximal color occurring
infinitely often in it is
even (resp. odd) | and we say that a color $c$ \emph{occurs} in
this play if there is some $v_i$ with $C(v_i)=c$.

A  \emph{winning strategy} for player $i$ is a mapping $H_i:V^*\cdot V^i\to V$,
such that for every play $v_0,v_1,\ldots$, player
$i$ always wins if $v_{j+1}= H_i(v_0,\ldots,v_j)$ whenever $v_j\in V^i$.
In addition, $H_i$ is \emph{memoryless} if $H_i(v_0,\ldots,v_j)$ agrees
with $H_i(v_j)$ for every $j$.

\begin{theorem}[\cite{GH82,Zie98,Jur98}]
\label{Thm: Game-winning}
For a parity game $G$, from every location,
there is exactly one player having a winning strategy.
The problem of deciding the winner at a location
is in \up$\cap$\coup.
In addition, if a player has a winning strategy then she also
has a memoryless one from the same location.
\end{theorem}

We use $\winning_i(G)$ to denote the set
consisting of all locations from which player $i$ has a winning strategy.

Given a \papa $A=(Q,q,\delta,\Omega)$, a \emph{gadget} $D$ of $A$ is a finite
directed acyclic digram
 $(P,\gamma)$ where $P\subseteq Q$, $\gamma\subseteq P\times P$,
and for each $q\in P$:
\begin{enumerate}
 \item if $\delta(q)=q'$, then $q'\in P$ and $(q,q')\in\gamma$;
  \item if $\delta(q)=q_1\wedge q_2$ then $q_1,q_2\in P$, and
  $(q,q_1),(q,q_2)\in\gamma$;
  \item if $\delta(q)=q_1\vee q_2$ then there is some $i\in\{1,2\}$ such that
  $q_i\in P$ and $(q,q_i)\in\gamma$,
  \item  $q$ has no successor for the other cases.
\end{enumerate}
For convenience, we sometimes directly write $q\in D$ whenever $D=(P,\gamma)$
and $q\in P$.
We denote by $\calD(A)$ the set consisting of all gadgets of $A$.
Since we require that each \papa $A$ is well-structured, then
$\calD(A)$ must be a finite set.

Given a sequence of gadgets $D_1,D_2,\ldots$ such that $D_i=(P_i,\gamma_i)$, an
\emph{infinite path} within it is a sequence of states
$q_{1,1},\ldots,q_{1,\ell_1},q_{2,1},\ldots,q_{2,\ell_2},\ldots$ such that
each $(q_{i,j},q_{i,j+1})\in\gamma_i$ and $\delta(q_{i,\ell_i})=\Circle^{\sim p_i} q_{i+1,1}$
for some $p_i\in[0,1]$.
We say such an infinite path is \emph{even} (resp. \emph{odd}) if the maximal
color (w.r.t. $\Omega$) occurring infinitely often is even (resp. odd).

We say that a gadget $D=(P,\gamma)$ is \emph{incompatible} if there
exist $q_1,q_2\in P$ and $\delta(q_1)=a$, $\delta(q_2)=\neg a$ for some
$a\in\calA$; or there is some $q\in P$ with $\delta(q)=\bot$.
Otherwise, we say that $D$ is \emph{compatible}.

Let $D$ be a gadget and $\Gamma=\{D_1,\ldots,D_k\}$ be a set of gadgets,
we denote by $\Gamma\Vdash D$ if there exist
$k$ positive numbers $x_1,\ldots,x_k$ such that:
   $\sum_{i=1}^{k}x_i\leq 1$, and
   for each $q\in D$ with $\delta(q)=\Circle^{\sim p}q'$, we have
    $\sum_{q'\in D_i}x_i\sim p$.
We in what follows call $x_1,\ldots,x_k$ the \emph{enabling condition}.
Note that the relation $\Vdash$ could be decided by solving a linear system
of inequality.

According to automata theory, we may construct a deterministic
(word) parity automaton $\widetilde{A}=(\widetilde{Q},\widetilde{q},\widetilde{\delta},
\widetilde{\Omega})$ were
$\widetilde{\delta}:\widetilde{Q}\times\calD(A)\to\widetilde{Q}$ and $\widetilde{\Omega}$
is a total coloring function.
It takes a gadget
sequence as input, and accepts it if every gadget in it is compatible and every infinite path
within it  is even.

Then, we may create a parity game $G_A=(V_A,E_A,C_A)$ for the \papa $A$, in
detail:
\begin{itemize}
  \item $V_A=V^0_A\cup V^1_A$,
where $V^0_A=2^{\calD(A)\times\widetilde{Q}}$ and $V^1_A=\calD(A)\times\widetilde{Q}$.
\item $E_A=\{(\{(D_1,\widetilde{q_1}),\ldots,(D_k,\widetilde{q_k})\},(D_i,\widetilde{q_i}))\mid 1\leq i\leq
k\}\cup\\ \{((D,\widetilde{q}),\{(D_1,\widetilde{q_1}),\ldots,(D_k,\widetilde{q_k})\})\mid (D_1,\ldots,D_k)\Vdash D,\\
	\text{ and each } \widetilde{q_i}=\widetilde{\delta}(\widetilde{q},D_i)\}$.
\item $C_A(D,\widetilde{q})=\widetilde{\Omega}(\widetilde{q})$,
hence every location in $V^1_A$ is colored.
\end{itemize}

\begin{theorem}
\label{Thm: Emp2Game}
Let the \papa $A=(Q,q,\delta,\Omega)$, then $\lang(A)\neq\emptyset$ if and
only if there is some $D\in\calD(A)$ with $q\in D$ such that
$\{(D,\widetilde{\delta}(\widetilde{q},D))\}\in\winning_0(G_A)$.
\end{theorem}
\begin{proof}
$\Longrightarrow$) Suppose that there is some pointed Markov chain
$(M=(S,\pmat{T},L),s)\in\lang(A)$, then there exists some accepting run
$(T,\lambda)$ of $A$ on $(M,s)$.
\par We say a vertex $v$ of $T$ is a \emph{modal vertex} if $
\delta(\proj_1(\lambda(v)))$ is of the form $\Circle^{\sim p} q'$.
We  denote by $\depth{v}$ the \emph{modal depth}
of $v$, i.e., the number of modal vertices among the ancestors of $v$.
\par
From each vertex $v$ of $T$, we may obtain a set
of vertices, denoted as $\cls(v)$, which involves $v$ and all its descendants
with the same modal depth.  Since $A$ is well-structured,
then $\cls(v)$ must be a finite set. We also lift the notation by defining
$\cls{V}=\bigcup_{v\in V}\cls(v)$ for a finite vertex set $V$.
\par
In addition, each finite vertex set $V$ of $T$ derives a gadget $\funcD(V)=(P_V,\gamma_V)$,
where $P_V=\{\proj_1(\lambda(v))\mid v\in V\}$, and $(q_1,q_2)\in\gamma_V$ if
there are two vertices $v_1,v_2\in V$, such that $\proj_1(\lambda(v_i))=q_i$
for $i=1,2$ and $v_2$ is a child of $v_1$.
\par
Let $v_0$ be the root vertex of $T$,
then we have $\lambda(v_0)=(q,s)$.
We now let $D=D_0= \funcD(\cls(v_0))$, then for each play
$\Delta_0,(D_0,\widetilde{q_0}),\Delta_1,(D_1,\widetilde{q_1}),\Delta_2,\ldots$
with $\Delta_0=(D,\widetilde{\delta}(\widetilde{q},D))$ and each $D_i=(P_i,\gamma_i)$,
player $0$ can control it and make the play to fulfill the following property:
\begin{description}
 \item[(*)] For each $i$, there exists a finite set of vertices $V_i$ having the same modal depth $i$,
 	and there exists a state $s_i$ of $M$;
         and $q'\in P_i$ iff there is some $v_{q'}\in V_i$ such that $\lambda(v_{q'})=(q',s_i)$.
 In addition, $(q_1,q_2)\in\gamma_i$ iff $v_{q_2}$ is a child of $v_{q_1}$.
\end{description}
For $i=0$,  we have $V_0=\cls(v_0)$ and $s_0=s$.
Assume that (*) holds at step $i$, then player $0$ chooses the next location
guided by the run as following:
First, let $V'_i$  be all modal vertices among $V_i$,
and let $V''_i$ be the set consisting of children of vertices in $V'_i$.
Then, $V''_i$ can be partitioned into several sets $V''_{i,1},\ldots,V''_{i,k}$
according to the second component (assume $\proj_2(\lambda(v'))=s_{i,j}$
for $v'\in V''_{i,j}$) labeled on the vertices.
Player $0$ then chooses the set
$\{(D_{i,1},\widetilde{q_{i,1}}),\ldots,(D_{i,k},\widetilde{q_{i,k}})\}$ as the
next location, where $D_{i,j}=\funcD(\cls(V''_{i,j}))$ and
$\widetilde{q_{i,j}}=\widetilde{\delta}(\widetilde{q_i},D_{i,j})$.
\par
Then, according to the construction, for each $(D_{i,j},\widetilde{\delta}(\widetilde{q_{i,j}}))$
we have some state $s_{i,j}$ and the vertex set $\cls(V''_{i,j})$
making  property (*) holds, no matter how player $1$ chooses.
Let $x_1=\pmat{T}(s_i,s_{i,1}),\ldots, x_k=\pmat{T}(s_i,s_{i,k})$, we definitely
have $\sum_{j=1}^{k}x_j\leq 1$ and we also have $\sum_{q''\in D_{i,j}}x_j\sim p$
for each $q'\in D_i$ such that $\delta(q')=\Circle^{\sim p}q''$ because
$(T,\lambda)$ is an accepting run. Therefore, $(D_{i,1},\ldots,D_{i,k})\Vdash
D_i$ holds.
\par
We assert that each $D_i=(q_1,\ldots,q_\ell)$ must be compatible
| since $(T,\lambda)$ is accepting, no such $q'\in D_i$ having $\delta(q')=\bot$,
and if there exist $q_1,q_2\in D_i$ with $\delta(q_1)=a$ and $\delta(q_2)=\neg a$,
then we will both have $a\in L(s_i)$ and $a\not\in L(s_i)$.
 Also note that each infinite path within
$D_0,D_1,\ldots$ corresponds to the first component of the labelings of an
infinite branch of $T$, hence it must be even.
According to $\widetilde{A}$, we then conclude that this
strategy is winning for player $0$ form $\{(D_0,\widetilde{\delta}(\widetilde{q}))\}$.
\par
$\Longleftarrow$) Let $H_0$ be the
(memoryless) winning strategy of player $0$ from
$\{(D,\widetilde{\delta}(\widetilde{q}))\}$, where $D$ is some gadget involving $q$.
We say that a location $l=(D^l,\widetilde{q^l})\in
V^1_A$ is \emph{feasible} if $l$ may appear in some play under control of player
$0$ according to $H_0$.
We  create  a Markov chain $M=(S,\pmat{T},L)$ as follows.
\begin{itemize}
  \item First, let $S=\{s_l\mid l \text{ is a feasible location}\}\cup\{s'\}$.
  \item Second, since each feasible location must be compatiable, then we may
  let $L(s_l)=\{a\in\calA\mid\text{there is some } q' \text{ in } D^l\}$.
  Meanwhile, we let $L(s')=\emptyset$.
  \item The transition matrix $\pmat{T}$ is determined as follows: For each
  feasible location $l$, suppose that
  $H_0(l)=\{l_1=(D^{l_1},\widetilde{q^{l_1}}),\ldots,l_k=(D_k,\widetilde{q^{l_k}})\}$,
  since  $\{D^{l_1},\ldots,D^{l_k}\}\Vdash D^l$ then
  we have a set of enabling condition $x_1,\ldots,x_k$.
We let $\pmat{T}(s_l,s_{l_j})=x_j$ for each $j$,
  let $\pmat{T}(s_l,s')=1-\sum_{j=1}^k x_j$, and let $\pmat{T}(s',s')=1$.
\end{itemize}
What left is to show that $(M,s_{l_0})\in\lang(A)$, where $l_0$ is just
$(D,\widetilde{\delta}(\widetilde{q}))$.
For each gadget $D^l$ such that $l$ is feasible, we could obtain
a forest $(T_l,\lambda_l)$, and in which each vertex $q'$ is
labeled with $(q',s_l)$.
Then from $T_{l_0}$ (which is an exact tree with $(q,s_{l_0})$ labeled in the root),
with a top-down manner, we
connect the so far added tree $T_l$ with every $T_{l'}$ such that $l'\in H_0(l)$
| i.e., for each $q'$ in $T_{l}$ with $\delta(q')=\Circle^{\sim p}q''$,
we add the vertex $q''$ in $T_{l'}$ as a child | it can be seen that it must
be the case that some edges connecting some leaves of $T_l$ and the root(s)
of $T_{l'}$. We denote the labeled tree finally get as $(T,\lambda)$, and
it is indeed be an accepting run of $A$ over $(M,s_{l_0})$.
\end{proof}

Intuitively, player $0$ could extract a winning strategy from
an accepting run of $A$ over any pointed Markov chain;
and conversely, one can construct a pointed Markov chain accepted by
$A$ according to the (memoryless) winning strategy of player 0.

As a consequence of Thm. \ref{Thm: Formula2Automaton},
Thm. \ref{Thm: Game-winning} and Thm. \ref{Thm: Emp2Game}
we have the following main conclusion of this section.

\begin{theorem}
\label{Thm: Decision}
Both the \emp problem of \papa and the \sat problem of \pmutl are decidable,
and both of them are in 2\exptime.
\end{theorem}

Indeed, from Thm. \ref{Thm: Formula2Automaton} one can get a \papa
whose scale is linear in the size of the input formula, and
an $n$-state \papa could be converted to a parity game with scale $2^{2^{\mathcal{O}(n)}}$.
From standard game theory (see \cite{Jur98,Wil00},
 and see \cite{Sch08} for an improved bound), and with
a similar analysis of \cite{Wil00} (see also the analysis of the coloring number
in that paper), one can infer that this problem is in 2\exptime.

\section{Discussion}
\label{Sec: Discussion}

In this paper, we present the logic \pmutl, a simple and succinct
probabilistic extension of \mutl.
We have compared the expressiveness of these two kinds of logics:
In general, \pmutl captures `local' and `stepwise' probabilities;
whereas \pctl could describe `global' probabilities  in
the system.
Hence, these two logics are orthogonal and complementary,
and one can obtain a more powerful and expressive
logic by combing them together, as done in \cite{CKP15}. i.e., we may use formulae  like
$(\mu Z.(a\vee\opX^{\geq 0.8}Z))\opU^{\geq 0.6}(\nu Z'.(b\wedge\opF^{>0.3}Z'))$.
Model checking algorithm of such an extension can be acquired from
those of the underlying logics.

In this paper, we have also investigated the decision problem of
\pmutl,  the key issue and the most challenging part
is  to deal with probabilistic quantifiers
when doing reduction to parity games,
which is a highly nontrivial extension of the non-probabilistic case.
As a cost, we have only now got an algorithm with double-exponential
time complexity for solving it |
in contrast, the \sat problem  for
the standard \mutl is in \exptime.


\section*{Acknowledgement} First and foremost, the authors would thank all
the anonymous reviewers for the valuable and helpful comments on this paper.
We would also thank Nir Piterman for his valuable comments on our work.

Wanwei Liu is supported by National Natural
Science Foundation of China (Grant Nos. 61103012, 61379054 and 61272335).
Lei Song is supported by Australian Research Council under Grant
DP130102764. Ji Wang is supported by National Natural
Science Foundation of China (Grant No. 61120106006).
Lijun Zhang (corresponding author) is supported by National Natural
Science Foundation of China (Grant Nos. 61428208, 61472473 and
61361136002), the CAS/SAFEA International Partnership Program
for Creative Research Teams.


\begin{thebibliography}{10}

\bibitem{BK08} C.~Baier and J.-P. Katoen.  \newblock {\em Principles
of Model Checking}.  \newblock MIT Press, 2008.

\bibitem{BB87} B.~Banieqbal and H.~Barringer.  \newblock Temporal
logic with fixed points.  \newblock In {\em B. Banieqbal,
H. Barringer, and A. Pnueli, editors, Temporal Logic in
Specification}, volume 398 of {\em Lecture Notes in Computer Science},
pages 62--74. Springer-Verlag, 1987.

\bibitem{Ber02} S.~Berezin.  \newblock {\em Model Checking and Theorem
Proving: A Unified Framework}.  \newblock Phd thesis, Carnegie Mellon
University, Pittsburgh, PA, USA, Jan.  2002.

\bibitem{BFS12} N.~Bertrand, J.~Fearnley, and S.~Schewe.  \newblock
Bounded satisfiability for {PCTL}.  \newblock In {\em {CSL}},
volume~16 of {\em {LIPIcs}}, pages 92--106. Schloss Dagstuhl -
Leibniz-Zentrum fuer Informatik, 2012.

\bibitem{bhatia2011motion} A.~Bhatia, M.~R. Maly, L.~E. Kavraki, and
M.~Y. Vardi.  \newblock Motion planning with complex goals.  \newblock
{\em Robotics \& Automation Magazine, IEEE}, 18(3):55--64, 2011.

\bibitem{BFKK08} T.~Br{\'{a}}zdil, V.~Forejt,
J.~Kret{\'{\i}}nsk{\'{y}}, and A.~Kucera.

\bibitem{CKP15} P.~Castro, C.~Kilmurray, and N.~Piterman.  \newblock
Tractable probabilistic $\mu$-calculus that expresses probabilistic
temporal logics.  \newblock In {\em 32nd Symposium on Theoretical
Aspects of Computer Science}, volume~30 of {\em Leibniz International
Proceedings in Informatics}, pages 211--223, 2015.

\bibitem{CIN05} R.~Cleaveland, S.~Iyer, and M.~Narasimha.  \newblock
Probabilistic temporal logics via the modal $\mu$-calculus.  \newblock
{\em Theor. Comput. Sci.}, 342(2-3):316--350, 2005.

\bibitem{CSS03} J.-M. Couvreur, N.~Saheb, and G.~Sutre.  \newblock An
optimal automata approach to {LTL} model checking of probabilistic
systems.  \newblock In {\em LPAR 2003}, volume 2850 of {\em Lecture
Notes in Computer Science}, pages 361--375. Springer, 2003.

\bibitem{Dam95} M.~Dam.  \newblock Translating {CTL*} into the modal
$\mu$-calculus.  \newblock Technical Report ECS-LFCS-90-123,
Laboratory for Foundations of Computer Science, University of
Edinburgh, November 1995.

\bibitem{AM01} L.~de~Alfaro and R.~Majumdar.  \newblock Quantitative
solution of omega-regular games.  \newblock In {\em {STOC}}, pages
675--683. {ACM}, 2001.

\bibitem{EC80} E.~Emerson and E.~Clarke.  \newblock Characterizing
correctness properties of parallel programs using fixpoints.
\newblock In {\em Proc. of the 7th Int. Colloquium on Automata,
Languages and Programming (ICALP'80)}, volume~85 of {\em Lecture Notes
in Computer Science}, pages 169--181. Springer-Verlag, 1980.

\bibitem{EH86} E.~A. Emerson and J.~Y. Halpern.  \newblock Decision
procedures and expressiveness in the temporal logics of branching
time.  \newblock {\em Journal of the ACM}, 33(1):151--178, 1986.

\bibitem{EL86} E.~A. Emerson and C.~L. Lei.  \newblock Efficient model
checking in fragments of the propositional mu-calculus.  \newblock In
{\em First IEEE Symposium on Logic in Computer Science}, pages
267--278. Los Alamitos: IEEE Computer Society, 1986.

\bibitem{GH82} Y.~Gurevich and L.~Harrington.  \newblock Trees,
automata, and games.  \newblock In {\em Proceeding of 14th ACM
Symposium on the Theory of Computing}, pages 60--65, San Francisco,
California, 1982.

\bibitem{HJ94} H.~Hansson and B.~Jonsson.  \newblock A logic for
reasoning about time and reliability.  \newblock {\em FAC},
6(5):512--535, 1994.

\bibitem{HK97} M.~Huth and M.~Z. Kwiatkowska.  \newblock Quantitative
analysis and model checking.  \newblock In {\em {LICS}}, pages
111--122. {IEEE} Computer Society, 1997.

\bibitem{Jur98} M.~Jurdzi\'{n}ski.  \newblock Deciding the winner in
parity games is in $\mathbf{UP}\cap$co-$\mathbf{UP}$.  \newblock {\em
Information Processing Letters}, 68(3):119--124, 1998.

\bibitem{Kat98} J.-P. Katoen.  \newblock {\em Concepts, Algorithms,
and Tools for Model Checking}.  \newblock FAU, Lehrstuhl f\"ur
Informatik VII Friedrich-Alexander Universit\"at Erlangen-N\"urnberg,
2 edition, 1998.  \newblock Lecture Notes of the Course ``Mechanised
Validation of Parallel Systems''.

\bibitem{Koz83} D.~Kozen.  \newblock Results on the propositional
$\mu$-calculus.  \newblock {\em Theoretical Computer Science},
27:333--354, 1983.


\bibitem{MM02} A.~McIver and C.~Morgan.  \newblock Games, probability
and the quantitative {\(\mathrm{\mu}\)}-calculus qm{\(\mathrm{\mu}\)}.
\newblock In {\em {LPAR}}, volume 2514 of {\em Lecture Notes in
Computer Science}, pages 292--310. Springer, 2002.

\bibitem{MM07} A.~McIver and C.~Morgan.  \newblock Results on the
quantitative $\mu$-calculus qm$\mu$.  \newblock {\em {TOCL}}, 8(1):3,
2007.

\bibitem{Mio12a} M.~Mio.  \newblock {\em Game semantics for
probabilistic modal $\mu$-calculi}.  \newblock PhD thesis, The
University of Edinburgh, 2012.

\bibitem{Mio12b} M.~Mio.  \newblock Probabilistic modal $\mu$-calculus
with independent product.  \newblock {\em Logical Methods in Computer
Science}, 8(4), 2012.

\bibitem{MM97} C.~Morgan and A.~McIver.  \newblock A probabilistic
temporal calculus based on expectations.  \newblock In L.~Groves and
S.~Reeves, editors, {\em Proc. Formal Methods Pacific}, pages
4--22. Springer, 1997.

\bibitem{Pnu77} A.~Pnueli.  \newblock The temporal logic of programs.
\newblock In {\em Proc. of 18th IEEE Symposium on Foundation of
Computer Science (FOCS' 77)}, pages 46--57. IEEE Computer Society,
1977.

\bibitem{Sch08} S.~Schewe.  \newblock {\em Synthesis of Distributed
Systems}.  \newblock Phd thesis, Saarbr{\"u}cken, 2008.

\bibitem{DBLP:journals/logcom/Singh98} M.~P. Singh.  \newblock
Applying the mu-calculus in planning and reasoning about action.
\newblock {\em J. Log. Comput.}, 8(3):425--445, 1998.

\bibitem{SC85} A.~P. Sistla and E.~M. Clarke.  \newblock The
complexity of propositional linear temporal logics.  \newblock {\em
Journal of Assoc. Comput. Mach.}, 32(3):733--749, 1985.

\bibitem{Var85} M.~Y. Vardi.  \newblock Automatic verification of
probabilistic concurrent finite-state programs.  \newblock In {\em
FOCS}, pages 327--338. IEEE Computer Society, 1985.

\bibitem{Wal00} I.~Walukiewicz.  \newblock Completeness of {K}ozen's
axiomatization of the propositional $\mu$-calculus.  \newblock {\em
Information and Computation}, 157:142--182, 2000.

\bibitem{Wil00} T.~Wilke.  \newblock Alternating tree automata, parity
games, and modal $\mu$-calculus.  \newblock {\em Bull, Belg, Math,
Soc}, 8(2):359--391, 2002.

\bibitem{Wol83} P.~Wolper.  \newblock Temporal logic can be more
expressive.  \newblock {\em Information and Control}, 56(1--2):72--99,
1983.

\bibitem{Zie98} W.~Zielonka.  \newblock Infinite games on finitely
coloured graphs with applications to automata on infinite trees.
\newblock {\em Theoretical Computer Science}, 200:135--183, 1998.

\end{thebibliography}

\end{document}